\documentclass[12pt]{amsart}
\usepackage{xcolor}
\usepackage{soul, url}
\usepackage{amsmath, amsfonts, amssymb, amsthm}
\usepackage{hyperref, comment}
\usepackage{graphicx} 
\usepackage[margin=1.5in]{geometry}
\usepackage{amsmath, amsthm, latexsym, amssymb, amsfonts, epsf, xcolor, hyperref, mathrsfs}

\newcommand{\F}{\mathbb{F}}

\newcommand{\Z}{\mathbb{Z}}
\renewcommand{\H}{\mathcal{H}}
\newcommand{\I}{\mathcal{I}}

\newcommand{\rmv}[1]{}

\usepackage{xcolor}
\usepackage{url}
\usepackage{wrapfig}
\usepackage{float}
\usepackage{geometry}
\usepackage{amssymb}
\usepackage{amsfonts}
\usepackage{amsthm}
\usepackage{mathtools}
\usepackage{enumerate}
\usepackage{verbatim}
\usepackage{physics}
\usepackage{euscript}
\usepackage{tabularx}
\usepackage[ruled,vlined,linesnumbered]{algorithm2e}
\usepackage[utf8]{inputenc}
\numberwithin{equation}{section}
\usepackage{enumitem}

\newtheorem{theorem}{Theorem}[section]
\newtheorem{lemma}[theorem]{Lemma}
\newtheorem{proposition}[theorem]{Proposition}

\theoremstyle{definition}
\newtheorem{definition}[theorem]{Definition} 
\newtheorem{remark}[theorem]{Remark}
\newtheorem{example}[theorem]{Example}

\title[Permutation decoding of algebraic geometry codes]{Permutation decoding of algebraic geometry codes from Hermitian and norm-trace curves}

\author[M. Lichtenwalner]{Monica Lichtenwalner}
\address[Monica Lichtenwalner]{Department of Mathematics\\ Virginia Tech\\ 
Blacksburg, VA, USA}
\email{mlichtenwalner@vt.edu}

\author[H. H. L\'opez]{Hiram H. L\'opez}
\address[Hiram H. L\'opez]{Department of Mathematics\\ Virginia Tech\\ 
Blacksburg, VA, USA}
\email{hhlopez@vt.edu}

\author[G. L. Matthews]{Gretchen L. Matthews}
\address[Gretchen L. Matthews]{Department of Mathematics\\ Virginia Tech\\ 
Blacksburg, VA, USA}
\email{gmatthews@vt.edu}

\author[P. Seneviratne]{Padmapani Seneviratne}
\address[Padmapani Seneviratne]{Department of Mathematics\\ Texas A\&M University-Commerce\\
Commerce, TX, USA}
\email{padmapani.seneviratne@tamuc.edu}

\thanks{Monica Lichtenwalner was partially supported by 
partially supported by NSF DMS-2201075.
Hiram H. L\'opez was partially supported by the NSF Grants DMS-2401558 and DMS-2502705. Gretchen L. Matthews was partially supported by NSF DMS-2201075, NSF DMS-2502705, and the Commonwealth Cyber Initiative.}



\begin{document}

\begin{abstract}
Permutation decoding is a process that utilizes the permutation automorphism group of a linear code to correct errors in received words. Given a received word, a set of automorphisms, called a PD set, moves errors out of the information positions so that the original message can be determined. In this paper, we investigate permutation decoding for certain families of algebraic geometry codes. Automorphisms of the underlying curve are used to specify permutation automorphisms of the code.  Specifically, we describe permutation decoding sets that correct specific burst errors for one-point codes on Hermitian and norm-trace curves.
\end{abstract}

\maketitle

\section{Introduction}
Prange~\cite{Prange} introduced permutation decoding for cyclic codes in 1962, and MacWilliams extended the concept to linear codes in 1964~\cite{MacWilliams_PD}. See also other early work~\cite{Shiva_Fung_Tan_PD_double_error_cyclic_70, Shiva_Fung_PD_triple_error_72}. Research continued with a focus on minimal sets of code automorphisms that support permutation decoding~\cite{Gordon} with particular attention to the Golay code~\cite{Wolfmann}. Permutation decoding of abelian codes was considered by Chabanne in 1992~\cite{Chabanne_92}, and then for group codes in~\cite{Biglieri_group_codes_95}. 
Partial permutation decoding was introduced later \cite{partial_Key_05, Kroll_Vincenti_05} to address the correction of all sets of $r$ errors with $r \leq t$, where $t$ denotes the number of correctable errors of the code.  Permutation decoding has been considered for codes from graphs, finite geometries, and designs \cite{Key_PD_finite_geometries_06, key_PD_lattice07}, first- and second-order Reed-Muller codes \cite{Bernal_new_advances_first_order_RM_23, SENEVIRATNE20091967}, and some images of Reed-Solomon codes \cite{Lim_perm_decoding_images_RS_10}.  

Permutation decoding uses a set of automorphisms, called a PD set, of a linear code $C$ to move errors in received words out of the information positions. Hence, codes with many automorphisms are well-suited to this decoding procedure. Much of the prior work considers codes that inherit automorphisms from the underlying structure that defines them. This motivates consideration of algebraic geometry codes defined over curves over finite fields. For a curve $\mathcal X$ of genus $g \geq 2$ over the finite field $\F_q$ with $q$ elements, its automorphism group $Aut(\mathcal X)$ satisfies
$$
\mid Aut(\mathcal X) \mid \leq 16g^4,
$$
which can be significantly larger than the Hurwitz bound 
$$\mid Aut(\mathcal X') \mid \leq 84(g-1)$$
for curves $\mathcal X'$ over fields of characteristic $0$; except for the Hermitian curve defined by
$$
\mathcal X_q: y^q+y=x^{q+1}
$$
which, according to \cite{Stichtenoth_1973}, has an even larger automorphism group given by $$\mid Aut(\mathcal X_q) \mid=16g^4 + 4q^7-7q^6+5q^5-q^4-q^3$$ and genus $g=\frac{q(q-1)}{2}$.

In this paper, we consider permutation decoding of algebraic geometry codes on norm-trace and Hermitian codes to handle specific burst errors. We provide explicit PD sets for these errors. For algebraic geometry codes of positive genus, there has been little other progress beyond Joyner's conjectures regarding permutation decoding for codes from hyperelliptic curves \cite{Joyner_conjectural_PD_05}. Reed-Solomon codes (which are algebraic geometry codes on the projective line, a curve of genus $0$) are themselves cyclic, so their permutation decoding is a consequence of MacWilliams' work \cite{MacWilliams_PD} as well as that of Prange \cite{Prange}. 
Our approach differs from that of Ren~\cite{Ren_03}, whose goal was to reduce the decoding complexity for burst errors in Hermitian codes. It applies to a larger family of codes and harnesses the automorphism groups of the code, applying it to the algebraic geometry codes directly. Moreover, the codes can correct random errors as given by the minimum distance. In contrast, the modification by Ren, which optimizes the Hermitian code representation for burst error correction, limits its ability to handle random errors. We determine  PD sets for correcting certain burst errors of one-point codes on the Hermitian and norm-trace curves. We also emphasize the role of information sets for these codes. Code automorphisms are of independent interest. They now support a number of applications, including fault-tolerant quantum 
computation \cite{Grassl_Roetteler, quantum_sayginel_25}; 
polar coding \cite{bioglio,  1Johannsen_polar, polar_19,Pillet_polar}, proving that some codes achieve capacity \cite{RM_CA1, RM_CA2}; and code-based cryptography \cite{LESS_is_even_more, LESS}.

This paper is organized as follows. The current section concludes with an overview of the notation to be used throughout the paper. Section \ref{S:prelim} contains necessary preliminaries on permutation decoding, algebraic geometry codes, and code automorphisms. Section \ref{S:herm} focuses on permutation decoding of Hermitian codes and structured information sets. Section \ref{S:nt} extends permutation decoding to norm-trace codes with information sets given by sets of arbitrary rational points. The paper ends with a summary and concluding remarks in Section \ref{S:concl}. Examples are contained in each section to demonstrate the results.

{\bf{Notation.}}
For a positive integer $n$, set $[n]:=\{1,\dots,n\}$, an denote the symmetric group on $n$ symbols by $S_n$. For $I \subseteq [n]$ and $\sigma\in S_n$, we write $\sigma \left( I \right) := \left\{ \sigma(i): i \in I \right\}$. 
The set of all $m \times n$ matrices with entries in the finite field $\F_q$ with $q$ elements is denoted $\F_q^{m \times n}$. We take $\F_q^n:=\F_q^{1 \times n}$ and let $e_1, \dots, e_n$ denote standard basis vectors so that $e_i:=(0, \dots, 0, 1, 0, \dots, 0)$ with $i$ being its only non-zero coordinate. The $n \times n$ identity matrix is denoted by $I_n \in \F_q^{n \times n}$. The weight of a vector $v \in \F_q^n$ is $wt(v):=\mid \left\{ i \in [n]: v_i \neq 0 \right\} \mid$.

An $[n,k,d]$ code $C$ over $\F_q$ is an $\F_q$-subspace of $\F_q^n$ with dimension $k$ and minimum distance $d$, meaning the Hamming distance 
$d(u,v):=\mid \left\{ i \in [n]: u_i \neq v_i \right\} \mid$ between $u,v \in \F_q^n$ satisfies $d(c,c') \geq d$ for any $c,c' \in C$. Such a code $C$ is called $t$-error correcting code, with $t \leq \left\lfloor \frac{d-1}{2} \right \rfloor$. We also say that $C$ is an $[n,k]$ code when the distance parameter $d$ is not relevant in the context. A generator matrix for $C$ is any matrix $G \in \F_q^{k \times n}$ whose rows form a basis for $C$. The dual of $C$ is an $[n,n-k,d']$ code over $\F_q$ given by 
$C^{\perp}:=\left\{ w \in \F_q^n: wc^T = 0 \ \forall c \in C\right\}$. A parity check matrix for $C$ is any matrix that is a generator matrix for $C^{\perp}$, equivalently any matrix $H \in \F_q^{n-k \times n}$ such that $Hc^T=0$ for all $c \in C$.
A set of $k$ coordinates $I \subseteq [n]$ is an information set for $C$ provided the associated set of columns of $G$ is linearly independent. Given $y \in \F_q^n$ and $I:=\{i_1, \dots, i_l\} \subseteq [n]$, let $y\mid_I:=(y_{i_1}, \dots, y_{i_l}) \in \F_q^l$.

\section{Preliminaries} \label{S:prelim}

In this section, we present the background for the main results presented in the subsequent sections. We discuss permutation decoding and review algebraic geometry codes and their permutation automorphism groups. 
\subsection{Permutation decoding}

Given a code $C \subseteq \F_q^n$ and $ \sigma \in S_n$, 
$$\sigma(c)=\left( c_{\sigma^{-1}(1)}, c_{\sigma^{-1}(2)}, \dots, c_{\sigma^{-1}(n)} \right).$$
 The permutation automorphism group of the code $C$ is $$Aut(C):= \left\{ \sigma \in S_n : \sigma(C)=C \right\}.$$
In this paper, we say code automorphism to mean an element of the permutation automorphism group of the code, since we do not consider monomial automorphisms.
 \rmv{Equivalently, we can define this group as $$
Aut(C):=\left\{ P \in \F_q^{n \times n} : c P \in C \ \forall c \in C, P \textnormal{ is a permutation matrix} \right\}.
$$}
 
Consider an $[n,k,d]$ code $C$ over $\F_q$  with generator matrix $G=[I_k \mid A ] \in \F_q^{k \times n}$ and parity check matrix $H=[-A^T \mid I_{n-k}] \in \F_q^{n-k \times n}$. Then $I=[k]$ denotes the associated information set. Permutation decoding depends on the representation of the code, so it is necessary to fix a particular systematic form for the generator (and associated) parity-check matrix. For convenience, we take $G$ as above so that $I=[k]$. The method described below also applies to a generator matrix $G$ in standard form, meaning that $e_1, \dots, e_k \in \F_q^k$ are among its columns but not necessarily columns $1, \dots, k$ as in the systematic form. 

\begin{definition}
A subset $S \subseteq Aut(C)$ is an $r$-PD set for $C$ if for every set $\left\{ i_1, \dots, i_r \right\} \subseteq [n]$ of $r$ coordinate positions there is an element $\sigma \in S$ such that $\sigma \left( \left\{ i_1, \dots, i_r \right\} \right) \subseteq \left\{ k+1, \dots, n \right\}$.
If $r=t:=\left\lfloor \frac{d-1}{2} \right \rfloor$, then $S$ is called a PD set for $C$. If $r<t$, then $S$ is called a partial PD set for $C$. 
\end{definition}

\begin{lemma}\cite[Theorem 8.1]{Handbook_chapter} \label{L:syndrome_wt}
Consider an $[n,k,d]$ code $C$ over $\F_q$  with 
information set
$I=[k]$ and parity check matrix $H=[-A^T \mid I_{n-k}] \in \F_q^{n-k \times n}$. 
    For any received word $y \in \F_q^n$, $y\mid_I$ is correct (meaning there are no errors in the first $k$ coordinates of $y$) if and only if $wt(Hy^T) \leq t$ where $t=\left\lfloor \frac{d-1}{2} \right\rfloor.$
\end{lemma}

{\bf Permutation Decoding Procedure:} Given a PD set $S=\left\{ \varphi_1, \dots, \varphi_s \right\} \subseteq Aut(C)$ and a received word $y \in \F_q^n$:
\begin{enumerate}
    \item Find $i \in [s]$ such that 
    \begin{equation} \label{E:syndrome_wt}
        wt \left( H \varphi_i(y)^T \right) \leq t.
    \end{equation}
    \item Decode $y$ as 
$$\varphi_i^{-1}\left( \varphi_i(y)\mid_I G \right) \in C.$$
\end{enumerate}

Step (1) may be handled by computing $wt \left( H \varphi_i(y)^T \right)$ for all $i \in [s]$ until $i$ is found satisfying Equation \ref{E:syndrome_wt}. 
Then, according to Lemma \ref{L:syndrome_wt}, $\varphi_i(y) \mid_I$ has no errors. Hence, if $S$ is an $r$-PD set, then any $r$ errors may be corrected using the Permutation Decoding Procedure. \rmv{We can extend this method to use monomial automorphisms, taking $S \subseteq MAut(C)$.}

The worst-case time complexity for the permutation decoding algorithm using an
$r$-PD set $S$ for an $[n,k]$ code is $O(nk|S|)$, which motivates the desire for small PD sets. Of course, not all codes have even partial PD sets. In fact,
$$
|S| \geq \left \lceil \frac{n}{n-k} \left \lceil \frac{n-1}{n-k-1}
\cdots \left \lceil \frac{n-r+1}{n-k-r+1} \right \rceil \cdots
\right \rceil \right \rceil;
$$
see \cite{Gordon,  Handbook_chapter, Schonheim_coverings_64}. 
Given that a code $C$ might not even have enough automorphisms to have a partial PD set, it is reasonable to consider sets of automorphisms that may be used to correct particular error patterns, as described in the following definition.

\begin{definition}
A subset $S \subseteq Aut(C)$ is an $r$-PD set for $E \subseteq [n]$ if for every set $\left\{ i_1, \dots, i_r \right\} \subseteq E$ of $r$ coordinate positions there is an automorphism $\sigma \in S$ such that $\sigma \left( \left\{ i_1, \dots, i_r \right\} \right) \subseteq \left\{ k+1, \dots, n \right\}$.
\end{definition}

It is immediate to see that an $r$-PD set is an $r$-PD set for $E=[n]$.

\subsection{Algebraic geometry codes}

Let $D=P_1+\dots+P_n$ and $G$ be divisors with  disjoint supports on a smooth projective curve
$\mathcal X$ over $\F_q$ so that $\deg G < n$  and $P_i \neq P_j$ for all $i,j \in [n]$ with $i \neq j$. 
Consider the algebraic geometry code $C(D,G):=\left\{ ev(f) : f \in \mathcal L(G) \right\}$
where $ev(f):=(f(P_1), \dots, f(P_n))$ and $\mathcal L(A):=\left\{ f \in \F_q(\mathcal X): (f) \geq -A \right\} \cup \{ 0 \}$ denotes the Riemann-Roch space of a divisor $A$. Recall that $C(D,G)$ is an $[n,\ell(G)-\ell(G-D), \geq n - \deg G]$ code where $\ell(G):=\dim_{\F_q} \mathcal L(G)$. Standard examples of algebraic geometry codes from curves of positive genus include the one-point Hermitian and norm-trace codes, which we describe below.

The Hermitian curve $\mathcal X_q:y^q+y=x^{q+1}$ is maximal over $\F_{q^2}$ and has genus $g=\frac{q(q-1)}{2}$. The set of affine $\F_{q^2}$-rational points of  $\mathcal X_q$ is
\[\mathcal X_q(\F_{q^2})=\bigcup_{a \in \F_{q^2}} \mathcal P_a,
\quad \text{where} \quad
\mathcal P_a:= \left\{ (a,b+\beta): \beta \in \ker (Tr) \right\}\]
for some $b \in \F_{q^2}$ satisfying $b^q+b=a^{q+1}$ and $Tr:\F_{q^2} \rightarrow \F_q$ denotes the trace map with respect to the extension $\F_{q^2}/\F_q$.
Similarly, the $\F_{q^2}$-rational points of $\mathcal X_q$ may be partitioned as
\[\mathcal X_q(\F_{q^2})=\bigcup_{b \in \F_{q^2}} \mathcal Q_b, \quad
\text{where} \quad
\mathcal Q_b:= \left\{ (\alpha a,b): \alpha \in \ker (N) \right\}
\quad
\]
for some $a \in \F_{q^2}$ satisfying $b^q+b=a^{q+1}$ and 
$N:\F_{q^2} \rightarrow \F_q$ denotes the norm map with respect to the extension $\F_{q^2}/\F_q$. In addition, $\mathcal X_q$ has a unique point at infinity $P_{\infty}:=(0:1:0)$.

The norm-trace curve over $\F_{q^s}$ is defined by
$$
\mathcal X_{q,s}:y^{q^{s-1}}+\dots + y^q+y=x^{\frac{q^s-1}{q-1}}
$$ where $s \geq 2$, meaning $Tr(y)=N(x)$ where 
$Tr: \F_{q^s} \to \F_q$ is the trace map with respect to the extension $\F_{q^s}/\F_q$.
The genus of $\mathcal X_{q,s}$ is $g=\frac{1}{2}\left( \frac{q^s-1}{q-1}-1\right) \left(q^{s-1}-1 \right)$ and 
 $\mathcal X_{q,s}$ has $q^{2s-1}+1$ $\F_{q^s}$-rational points, including a single point at infinity $P_{\infty}:=(0:1:0)$. When $s=2$, the norm-trace curve is a Hermitian curve. 
 
 When convenient, we write $P_{ab}:=(a,b)$ for affine points on these curves. Furthermore, we use $\I$ to refer to the set of points on a curve that are associated with information positions of a code, and we use $\H$ to refer to the set of points on a curve that are associated with check positions of a code.

\subsection{Code automorphisms from curve automorphisms}

Let $\mathcal X$ be a curve of genus $g$.  We will see that the code $C(D,G)$ inherits some of the automorphisms of the curve $\mathcal X$, noting that the code may have additional automorphisms that are not associated with automorphisms of $\mathcal X$.

Each automorphism of $\mathcal X$ that fixes $D$ and $G$,  meaning those in 
$$Aut_{D,G}(\mathcal X)=\left\{ \sigma \in Aut(\mathcal X): \sigma(D)=D, \sigma(G)=G\right\},$$ gives rise to a permutation automorphism of the code $C(D,G)$ as follows.
Given $\sigma \in Aut_{D,G}(\mathcal X)$, we may define 
$\overline{\sigma} \in S_n$ by 
$$
\begin{array}{cccc}
\overline{\sigma}: &[n] &\rightarrow &[n] \\
& i & \mapsto & j 
\end{array}
$$ 
if and only if $\sigma(P_i)=P_j$. If $\deg D > 2g-2$, the map
$$
\begin{array}{cccc}
\varphi: &Aut_{D,G}(\mathcal X) & \rightarrow & Aut \left( C(D, G) \right) \\
&\sigma & \mapsto & \overline{\sigma}
\end{array}
$$
is injective \cite[Proposition 8.2.3]{Stichtenoth}, and the identification $ \sigma \leftrightarrow  \overline{\sigma}$ gives
\begin{equation}\label{E:curve_code_autom}
Aut_{D,G}(\mathcal X) \leq Aut \left(C(D, G) \right).
\end{equation}
Wesemeyer \cite{Wesemeyer_98} gave conditions that guarantee equality in Expression (\ref{E:curve_code_autom}) for one-point codes on admissible curves. Recall that $\mathcal X$ is admissible if $g>0$, there are rational functions $x, y \in \F_q\left( \mathcal X \right)$ with pole divisors $(x)_{\infty}=l_1P_{\infty}$ and  $(y)_{\infty}=l_2 P_{\infty}$ for $l_1, l_2 \in Z^+$, and $\mathcal L(\gamma P_{\infty})$ has a basis given by $\left< x^i y^j : il_1+jl_2 \leq \gamma \right>$ for $\gamma \in \Z^+$ \cite[Definition 5.1]{Wesemeyer_98}. 

\begin{lemma} \cite[Theorem 5.12]{Wesemeyer_98} \label{L:Wesemeyer_admissible}
Suppose $\mathcal X$ is an admissible curve  over $\F_q$ Assume that $1\leq l_1<l_2\leq\gamma$. If 
\begin{equation} \label{E:Wesemeyer_condition}
    \deg D > \max \left\{ 2g+2, 2\gamma, l_1 \left( l_2 + \frac{l_1-1}{\beta} \right), l_1 l_2 \left( 1 + \frac{l_1-1}{\gamma-l_1+1} \right)\right\},
\end{equation}
where $\beta=\min\{l_1-1,r:y^r\in\mathcal L(\gamma P_\infty)\}$, then
  $$Aut_{D,\gamma P_{\infty}}(\mathcal X) = Aut \left(C(D, \gamma P_{\infty}) \right).$$
\end{lemma}

We note that Hermitian curves and more generally norm-trace curves are admissible, so Lemma \ref{L:Wesemeyer_admissible} applies to the associated one-point codes for which Equation (\ref{E:Wesemeyer_condition}) holds. The case for Hermitian curves is summarized in the following result.

\begin{lemma} \cite[Theorem 5.15, Corollary 5.16]{Wesemeyer_98}
Let $\mathcal X_q$ be the Hermitian curve over $\F_{q^2}$ and $D=\sum_{P\neq P_\infty}P$. Assume $q>2$. If 
$$q+1\leq\gamma\leq\frac{q^3-1}{2} \quad \text{ or } \quad \frac{q^3+1}{2}+q^2-q-2\leq\gamma\leq q^3+q^2-2q-3,$$
then
$$Aut(C(D,\gamma P_\infty))=Aut_{D,\gamma P_\infty}(\mathcal X_q).$$
\end{lemma}

A similar result holds for more general algebraic geometry codes of modest dimension, according to the following result. 

\begin{lemma} \cite[Corollary 3]{Joyner_Ksir_06} Suppose $\mathcal X$ is a curve of genus $g \geq 2$ with divisors $D$ and $G$ having disjoint support. If $\deg D \geq (1+g) \deg G$ and $\deg G \geq 2g+1$, then 
$$Aut_{D,G}(\mathcal X) = Aut \left(C(D, G) \right).$$
\end{lemma}

For permutation decoding, we do not necessarily require the entire automorphism group of the code. In the following sections, we will see that Expression \ref{E:curve_code_autom} provides $r$-PD sets for the sets of burst errors of interest. 

\section{Hermitian codes} \label{S:herm}

\subsection{Automorphisms}\label{S:herm:aut}

Automorphisms of the Hermitian curve were first studied by Stichtenoth in 1973 \cite{Stichtenoth_1973}. 
The following convenient matrix representation is taken from \cite[Lemma 2]{triples}. 

\begin{lemma}\label{L:Herm_autom}
Let $(a:b:e)$ and $(c:d:f)$ be any two distinct points of $ \mathcal X_q(\F_{q^2})$ and $\epsilon$ an element of $\F_{q^2}^*$.
Then there is an automorphism of $\mathcal X_q$ induced by the linear mapping on $\mathbb P^2$ defined by left multiplication by the matrix
\[
  M = \begin{bmatrix}
\epsilon(ed - b f)^q & \epsilon^{q+1} \xi c & a \\
\epsilon(a d - b c)^q & \epsilon^{q+1} \xi d & b \\
\epsilon(ec  - a f)^q & \epsilon^{q+1} \xi f & e 
\end{bmatrix} ,
\]
where $\xi = -c^q a + d^q e + f^q b$.
Moreover, every element of $Aut(\mathcal X_q)$ can be written in this form for some choice of $\epsilon$ and the two points.
\end{lemma}

Hence, 
$$Aut \left( \mathcal X_q \right) = \left\{ \varphi_{a,b,e,c,d,f,\epsilon} : (a:b:e), (c:d:f) \in \mathcal X_q(\F_{q^2}), \epsilon \in \F_{q^2}^* \right\}$$ where 
$$\varphi_{a,b,e,c,d,f,\epsilon}:
\begin{cases} 
x \mapsto \epsilon (ed-bf)^q x + \epsilon^{q+1} \xi c y+ az \\
y \mapsto \epsilon (ad-bc)^q x + \epsilon^{q+1} \xi d y + bz \\
z \mapsto \epsilon (ec-af)^q x + \epsilon^{q+1} \xi f y + ez. \\
\end{cases}$$
Note that $$\varphi_{a,b,e,c,d,f,\epsilon}(x:y:z)=M(x:y:z)^T.$$ We have  \[\varphi_{a,b,1,c,d,1,\epsilon}(P_{00})=P_{ab} \quad \text{ and } \quad \varphi_{a,b,1,c,d,1,\epsilon}(P_{\infty})=P_{cd}.\] One can also verify that \[\varphi_{a,b,e,c,d,f,\epsilon}(P_{00})=(a:b:e) \quad \text{ and } \quad \varphi_{a,b,e,c,d,f,\epsilon}(P_{\infty})=(c:d:f).\] 
Thus, given any two points $(a:b:e)$ and $(c:d:f)$ on $\mathcal X_q$, there are $q^2-1$ automorphisms $\varphi_{a,b,e,c,d,f,\epsilon}$ of $\mathcal X_q$ such that $\varphi_{a,b,e,c,d,f,\epsilon}(P_{00})=(a:b:e)$ and $\varphi_{a,b,e,c,d,f,\epsilon}(P_{\infty})=(c:d:f)$. From this description, we can conclude
that $|Aut \left( \mathcal X_q \right)|=(q^3+1)q^3(q^2-1)$. 

To determine $\Gamma:=\left\{ \sigma \in Aut \left( \mathcal X_q \right) : \sigma(P_{\infty})=P_{\infty} \right\}$, consider those $\varphi_{a,b,e,c,d,f,\epsilon}$ such that $\varphi_{a,b,e,c,d,f,\epsilon}(P_{\infty})=P_{\infty}$. Recall that
$\varphi_{a,b,e,c,d,f,\epsilon}(P_{\infty})=(c:d:f)$, so $(c:d:f)=(0:1:0)$ which forces $c=f=0$. The only such maps are 
$$\varphi_{a,b,e,0,d,0,\epsilon}:
\begin{cases} 
x \mapsto \epsilon (ed)^q x +  az \\
y \mapsto \epsilon (ad)^q x + \epsilon^{q+1} \xi d y + bz \\
z \mapsto ez \\
\end{cases}$$
where $\xi=d^qe.$ We can further simplify the description by noting that we may assume $d=e=1$. This gives 
\begin{equation}\label{E:phi_a_b_ep} \varphi_{a,b,1,0,1,0,\epsilon}:
\begin{cases} 
x \mapsto \epsilon x +  az \\
y \mapsto \epsilon  a^q x + \epsilon^{q+1}  y + bz \\
z \mapsto z, \\
\end{cases} \end{equation}
or equivalently, the linear map defined by left multiplication of the matrix
\[
  M = \begin{bmatrix}
\epsilon & 0 & a \\
\epsilon a^q & \epsilon^{q+1} & b \\
0 & 0 & 1
\end{bmatrix} ,
\]

We note that $| \Gamma | = q^3 (q^2-1)$ given that each $P_{ab}$ and $\epsilon \in \F_{q^2}^*$ gives rise to such a map.

Automorphisms of Hermitian codes $C(D,\gamma P_{\infty})$ were first studied by 
Xing in 1995 \cite{Xing_95}.
Indeed,
$$Aut_{D,G}(\mathcal X_q)=\left\{ \varphi_{a,b,\epsilon}: P_{ab} \in \mathcal X_q(\F_{q^2}), \epsilon \in \F_{q^2}^* \right\}$$ where
$
\varphi_{a,b,\epsilon}:=
\varphi_{a,b,1,0,1,0,\epsilon}$ as in Expression \ref{E:phi_a_b_ep} so that 
$$
\begin{array}{lccl}
\varphi_{a,b,\epsilon}: &x &\mapsto& \epsilon x + a \\
&y &\mapsto& 
\epsilon  a^q x + \epsilon^{q+1}   y + b 
\end{array}
$$
For maps of the form $\varphi_{a,b,1}$, we often write $\varphi_{a,b}$.

\subsection{Systematic Representation}\label{S:herm:systrep}

In this section, we see that automorphisms of the Hermitian code $C(D,\gamma P_{\infty})$ with large order partition the affine points of $\mathcal X_q$ into a small number of orbits. These orbits provide a description of a systematic form generator matrix for $C(D,\gamma P_{\infty})$.

Consider the automorphism
$$\sigma:
\begin{cases} x \mapsto \zeta x \\ y \mapsto \zeta^{q+1}y
\end{cases}$$
where $\F_{q^2}^*= \left< \zeta \right>$; that is, $\sigma=\varphi_{0,0,\zeta}$. Then there are $q+2$ orbits of the affine points on $\mathcal X_q$:
\begin{itemize}
\item $q$ orbits of the form
$
O_i:= \left\{ P_{\zeta b_{1}}, P_{\zeta^2 b_2}, P_{\zeta^3 b_3},  \dots,  P_{1,b_{q^2-1}} \right\},
$
\item
$O_{q+1}:=\left\{ P_{0b} : b \neq 0 \right\}$, and
\item
$O_{q+2}:=\left\{ P_{00} \right\}$.
\end{itemize}
In particular, if we fix $b \in \F_{q^2}$ so that $(\zeta:b) \in \mathcal X_q(\F_{q^2})$ and write
$\ker (Tr)=\left\{ \beta_1, \dots, \beta_q \right\}$ with $\beta_1=0$, then 
for $i \in [q]$, 
\begin{equation}
O_i=\left\{ (\zeta^t, \zeta^{(t-1)q+t}b+\zeta^{(t-1)(q+1)}\beta_i): t \in [q^2-1] \right\}\label{orbit order}
\end{equation}
and 
$$
O_{q+1}=\left\{ (0,\beta_i): i \in [q] \right\}.
$$
Observe that if $(\zeta^t,\zeta^\ell)\in O_j$, then $$\sigma^n(\zeta^t,\zeta^\ell)=(\zeta^{t+n},\zeta^{n(q+1)+\ell})=(\zeta^{t+n},\zeta^\ell)\in O_j$$ for all $n\in\{k(q-1):k\in[q+1]\}$. Since $|\{k(q-1):k\in[q+1]\}|=q+1=|\mathcal Q_{\zeta^\ell}|$, it follows that $\mathcal Q_{\zeta^\ell}\subseteq O_j$ for some $j\in[q]$. Thus, for any $b\in\F_{q^2}$, all points in $\mathcal Q_b$ belong to the same orbit of $\sigma$.

The orbits of $\sigma$ yield a description for systematic generator matrices for one-point Hermitian codes $C(D,\gamma P_{\infty})$ as given in \cite{Little_Saints_Heegard_97}; see also \cite{Heegard_Little_Saints_95}.

\begin{proposition}\label{prop:systmetic}
    \cite[Theorems 3.3 and 3.4]{Little_Saints_Heegard_97}\label{info_positions}
 Let $C(D,\gamma P_\infty)$ be a Hermitian code with $O_1,...,O_q$ specified as in Expression (\ref{orbit order}). If $\gamma=(i+1)(q^2-1)-q+j$ with $0 \leq j \leq q-1$,
 then one can take
\begin{enumerate}
\item as information positions all points in $O_1, \dots, O_i$ and
the first $\frac{q(q-1)}{2}+j$ points in $O_{i+1}$; and \item as check
positions all points in $O_{i+2}, \dots, O_q, O_{q+1}, O_{q+2}$
and $q^2-1-\left(\frac{q(q-1)}{2}+j \right)$ points in $O_{i+1}$.
\end{enumerate}
If $\gamma=i(q^2-1)+tq+r$ with $0 \leq r \leq q-1$, then one can take
\begin{enumerate}
\item as information positions all points in $O_1, \dots,
O_{i-1}$, the first $t+\dots+q-1+r-(t+1)$ points in $O_i$, and the first
$1+\dots+q-t+r$ points in $O_{i+1}$; and \item as check positions
all points in $O_{i+2}, \dots, O_q, O_{q+1}, O_{q+2}$, the remaining
\\$q^2-1-\left(t+\dots+q-1+r-(t+1) \right)$ points in $O_{i}$, and the remaining \\
$q^2-1-\left( 1+\dots+q-t+r \right)$ points in $O_{i+1}$.
\end{enumerate}
\end{proposition}

As we make use of the organization of rational points provided in Proposition \ref{info_positions}, we include the following examples to illustrate it. 

\begin{example} \label{E:q=4_orbits}
    Consider $\mathcal X_4:y^4+y=x^5$ over $\F_{16}$. Let $\F_{16}^* = \left< \zeta \right>$ where $\zeta^4=\zeta+1$.
Consider the map
$$
\begin{array}{lccc}
\sigma:&x &\mapsto &\zeta x \\
&y&\mapsto&\zeta^5y.
\end{array}
$$
Then 
$$
\begin{array}{lcl}
O_1&=&\left\{ (\zeta^t,\zeta^{5t+1}): t \in [15] \right\}=\mathcal Q_{\zeta}\cup\mathcal Q_{\zeta^6}\cup\mathcal Q_{\zeta^{11}}\\
O_2&=&\left\{ (\zeta^t,\zeta^{5t+2}): t \in [15] \right\}=\mathcal Q_{\zeta^2}\cup\mathcal Q_{\zeta^7}\cup\mathcal Q_{\zeta^{12}}\\
O_3&=&\left\{ (\zeta^t,\zeta^{5t+8}): t \in [15] \right\}=\mathcal Q_{\zeta^3}\cup\mathcal Q_{\zeta^8}\cup\mathcal Q_{\zeta^{13}}\\
O_4&=&\left\{ (\zeta^t,\zeta^{5t+4}): t \in [15] \right\}=\mathcal Q_{\zeta^4}\cup\mathcal Q_{\zeta^9}\cup\mathcal Q_{\zeta^{14}}\\
O_5&=&\left\{ (0,\zeta^5), (0,\zeta^{10}),(0,1)\right\}=\mathcal Q_{\zeta^5}\cup\mathcal Q_{\zeta^{10}}\cup\mathcal Q_{\zeta^{15}}\\
O_6&=&\left\{(0,0) \right\}=\mathcal Q_{0}.\\
\end{array}
$$
\end{example}

\begin{example} \label{E:q=5_orbits}
    Consider $\mathcal X_5:y^5+y=x^6$ over $\F_{25}$. Let $\F_{25}^* = \left< \zeta \right>$ where $\zeta^2-\zeta+2=0$. Then $\ker(Tr)=\left\{0,\zeta^3,\zeta^9,\zeta^{15},\zeta^{21} \right\}$ and $\ker(N)=\left\{ 1, \zeta^4, \zeta^8, \zeta^{12},\zeta^{16},\zeta^{20} \right\}$.
Consider the map
$$
\begin{array}{lccc}
\sigma:&x &\mapsto &\zeta x \\
&y&\mapsto&\zeta^6y.
\end{array}
$$
Then 
$$
\begin{array}{lcl}
O_1&=&\{(\zeta^t,\zeta^{6t+1}):t\in[24]\}=\mathcal Q_{\zeta}\cup\mathcal Q_{\zeta^7}\cup\mathcal Q_{\zeta^{13}}\cup\mathcal Q_{\zeta^{19}}\\
O_2&=&\{(\zeta^t,\zeta^{6t+18}):t\in[24]\}=\mathcal Q_{\zeta^6}\cup\mathcal Q_{\zeta^{12}}\cup\mathcal Q_{\zeta^{18}}\cup\mathcal Q_{\zeta^{24}}\\
O_3&=&\{(\zeta^t,\zeta^{6t+4}):t\in[24]\}=\mathcal Q_{\zeta^4}\cup\mathcal Q_{\zeta^{10}}\cup\mathcal Q_{\zeta^{16}}\cup\mathcal Q_{\zeta^{22}}\\
O_4&=&\{(\zeta^t,\zeta^{6t+5}):t\in[24]\}=\mathcal Q_{\zeta^5}\cup\mathcal Q_{\zeta^{11}}\cup\mathcal Q_{\zeta^{17}}\cup\mathcal Q_{\zeta^{23}}\\
O_5&=&\{(\zeta^t,\zeta^{6t+20}):t\in[24]\}=\mathcal Q_{\zeta^2}\cup\mathcal Q_{\zeta^8}\cup\mathcal Q_{\zeta^{14}}\cup\mathcal Q_{\zeta^{20}}\\
O_6&=&\left\{(0,\zeta^3), (0,\zeta^9),(0,\zeta^{15}),(0,\zeta^{21})\right\}=\mathcal Q_{\zeta^3}\cup\mathcal Q_{\zeta^{9}}\cup\mathcal Q_{\zeta^{15}}\cup\mathcal Q_{\zeta^{21}}\\
O_7&=&\left\{(0,0) \right\}=\mathcal Q_{0}.\\
\end{array}
$$
\end{example}

\subsection{Permutation Decoding}

We now demonstrate how permutation decoding can correct burst errors in one-point Hermitian codes. First, we use permutation decoding to correct burst errors in positions associated with points having the same $x$-value, as described in the following result.

\begin{theorem} \label{P:x_burst}
Let $C(D,\gamma P_{\infty})$ be a Hermitian code with $\gamma$ of the form of Proposition \ref{info_positions} such that $i\leq q$. For each $a \in \F_{q^2}$, fix $b_a$ such that $(a:b_a) \in \mathcal P_a$. Then 
$$
\left\{ \varphi_{-a,(-a)^{q+1}-b_a} : a \in \F_{q^2}\right\}
$$
is a partial $q$-PD set of size $q^2-1$ for burst errors impacting positions indexed by $\mathcal P_a$, $a \in \F_{q^2}$.
\end{theorem}

\begin{proof}
Fix $a\in\F_{q^2}$. We first show that $(-a,(-a)^{q+1}-b_a)\in\mathcal X_q(\F_{q^2})$. Suppose $a=0$. If $q$ is even, then $(-a,(-a)^{q+1}-b_a)=(a:b_a)\in\mathcal X_q(\F_{q^2})$. If $q$ is odd, then
$$(-b_a)^q+(-b_a)=-b_a^q-b_a=-(b_a^q+b_a)=-(a^{q+1})=0.$$
Next, suppose $a\neq 0$. Then
\begin{align*}
    \left((-a)^{q+1}-b_a\right)^q+(-a)^{q+1}-b_a&=(-a)^{q^2+q}+(-a)^{q+1}+(-b_a)^q+(-b_a)\\
    &=(-a)^{q^2-1}(-a)^{q+1}+(-a)^{q+1}+(-b_a)^q+(-b_a)\\
    &=2(-a)^{q+1}+(-b_a)^q+(-b_a).
\end{align*}
If $q$ is even, then
$$2(-a)^{q+1}+(-b_a)^q+(-b_a)=b_a^q+b_a=a^{q+1}=(-a)^{q+1}.$$
If $q$ is odd, then
\begin{align*}
    2(-a)^{q+1}+(-b_a)^q+(-b_a)&=2a^{q+1}-(b_a^q+b_a)\\
    &=2a^{q+1}-a^{q+1}\\
    &=a^{q+1}
    =(-a)^{q+1}.
\end{align*}
Thus, $(-a,(-a)^{q+1}-b_a)\in\mathcal X_q(\F_{q^2})$, and it follows that $\varphi_{-a,(-a)^{q+1}-b_a}\in Aut(C(D,\gamma P_\infty))$ for all $a\in\F_q^2$.

Next, observe that 
$$
\varphi_{-a,(-a)^{q+1}-b_a}((a,b_a))=(0,0) \in O_{q+2}
$$
and for all $\beta \in \ker (Tr)$,
$$
\varphi_{-a,(-a)^{q+1}-b_a}((a,b_a+\beta))=(0,\beta) \in O_{q+1}.
$$
Hence,
$$
\varphi_{-a,(-a)^{q+1}-b_a}\left( \mathcal P_a \right) \subseteq O_{q+1} \cup O_{q+2} \subseteq \H
$$
where $\H$ is the set of points on $\mathcal X_q$ associated with check positions of $C(D,\gamma P_{\infty})$.
\end{proof}

Using only the orbits $O_{q+1}$ and $O_{q+2}$, we are limited to correcting sets of errors of size at most $q$, since $\mid O_{q+1} \cup O_{q+2}\mid=q$.

\begin{example}
    Consider $\mathcal X_4:y^4+y=x^5$ over $\F_{16}$ as in Example \ref{E:q=4_orbits}. For notational purposes, label the $64$ points on $\mathcal X_4$ as $P_1,\ldots,P_{64}$. Suppose $y=c+e \in \F_{16}^{64}$ is received where $c=ev(f) \in C(D,\gamma P_{\infty})$ and $e=\lambda_1 e_1+\lambda_2 e_{16}+\lambda_3 e_{31}+\lambda_4 e_{46}$. Then 
    $$
    \begin{array}{ll}
\varphi_{1,1-a,1}(y)= 
    & ( f(P_{64}),f(P_{2'}), \dots, f(P_{15'}), 
    f(P_{61}), f(P_{17'}), \dots, f(P_{30'}), \\
     & f(P_{62}), f(P_{32'}), \dots, f(P_{45'}),
    f(P_{63}), f(P_{47'}), \dots, f(P_{60'}), \\ 
    & f(P_{1})+\lambda_1,  f(P_{16})+\lambda_2,
    f(P_{31})+\lambda_3,
    f(P_{46})+\lambda_4 )
    \end{array}
    $$
    where $P_{i'}=P_{\varphi_{1,1-a,1}^{-1}(i)}$ for $i \in [64]$.
We can see that $\varphi_{1,1-a,1}(y)\mid_{[60]}$ has no errors. Hence, $c':=\varphi_{1,1-a,1}(y) \mid_I G \in C(D,\gamma P_{\infty})$ and $c=\varphi_{1,1-a,1}^{-1}(c')\mid_I G$.
\end{example}

We next show that permutation decoding can correct burst errors in positions associated with points that share the same $y$-value.

\begin{theorem} \label{P:y_burst}
Given a Hermitian code $C(D,\gamma P_{\infty})$ where $\gamma$ is of the form of Proposition \ref{info_positions} such that $i\leq q-2$, a partial $(q^2-1)$-PD set of size $q$ for burst errors impacting positions indexed by $\mathcal Q_b$, $b \in \F_{q^2}$, is
$$
\left\{ \varphi_{0,\beta} :  \beta \in \ker(Tr) \right\}.
$$
{Moreover, these sets correct errors in any of the $q^2-1$ positions indexed by points in $\varphi_{0,\beta}^{-1}(O_q)$.}
\end{theorem}

\begin{proof}
If $b\in\ker(Tr)$, then $\varphi_{0,0}(\mathcal Q_b)=\mathcal Q_b=\{(0,b)\}\subseteq O_{q+1} \cup  O_{q+2}\subseteq[n]\backslash I$, where $I$ is an information set for $C(D,\gamma P_{\infty})$.

If $b\notin\ker(Tr)$, then fix $a \in \F_{q^2}^*$ such that $(a,b)\in\mathcal X_q$. There exists $\beta\in\ker(Tr)$ such that $(a,b+\beta) \in O_q$. Then
$\varphi_{0,\beta}(a,b)=(a,b+\beta) \in O_q$. Moreover, $\varphi_{0,\beta}(\alpha a,b)=(\alpha a,b+\beta) \in O_q$ for each $\alpha\in\ker(N)$. Thus,
$$\varphi_{0,\beta}(\mathcal Q_b)\in O_q\subset\H,$$
where $\H$ is the set of points on $\mathcal X_q$ associated with check positions of $C(D,\gamma P_{\infty})$.
\end{proof}

\begin{example}
    Consider $\mathcal X_4:y^4+y=x^5$ over $\F_{16}$ as in Example \ref{E:q=4_orbits}. For $b\in\ker(Tr)=\left\{ 0, \zeta^5, \zeta^{10}, 1 \right\}$, observe that
    $$\varphi_{0,0}(\mathcal Q_b)=\mathcal Q_b=\{(0,b)\}\subseteq O_{q+1} \cup  O_{q+2}.$$
    Hence, any errors in these positions will be corrected during permutation decoding. 

    For other positions, recall that 
$O_4=\mathcal Q_{\zeta^4} \cup \mathcal Q_{\zeta^9} \cup \mathcal Q_{\zeta^{14}}$.
    Moreover, 
    $$
\begin{array}{lcl} \varphi_{0,1}^{-1}(O_4)&=&\mathcal Q_{\zeta}\cup \mathcal Q_{\zeta^7} \cup \mathcal Q_{\zeta^3}, \\
\varphi_{0,\zeta^5}^{-1}(O_4)&=&\mathcal Q_{\zeta^2}\cup \mathcal Q_{\zeta^{13}} \cup \mathcal Q_{\zeta^{11}}, \\
\varphi_{0,\zeta^{10}}^{-1}(O_4)&=&\mathcal Q_{\zeta^8}\cup \mathcal Q_{\zeta^6} \cup \mathcal Q_{\zeta^{12}}, \\
\varphi_{0,0}^{-1}(O_4)&=&\mathcal Q_{\zeta^4}\cup \mathcal Q_{\zeta^9} \cup \mathcal Q_{\zeta^{14}}.
    \end{array}
    $$
    Therefore, $\left\{ \varphi_{0,0},
    \varphi_{0,\zeta^{5}}, \varphi_{0,\zeta^{10}}
    \right\}$ is a partial PD set for burst errors impacting positions indexed by $\mathcal Q_b$ for the $b$ listed above.
\end{example}

\begin{example}
    Consider $\mathcal X_5:y^5+y=x^6$ over $\F_{25}$ as in Example \ref{E:q=5_orbits}. For $b\in\ker(Tr)=\{0,\zeta^3,\zeta^9,\zeta^{10},\zeta^{21}\}$, observe that
    $$\varphi_{0,0}(\mathcal Q_b)=\mathcal Q_b=\{(0,b)\}\subseteq O_{q+1} \cup  O_{q+2}.$$
    Recall that 
$O_5=\mathcal Q_{\zeta^{20}} \cup \mathcal Q_{\zeta^2} \cup \mathcal Q_{\zeta^8} \cup \mathcal Q_{\zeta^{14}}$.
    Moreover, 
    $$
\begin{array}{lcl}
\varphi_{0,\zeta^{21},}^{-1}(O_5)&=&\mathcal Q_{\zeta}\cup \mathcal Q_{\zeta^6} \cup \mathcal Q_{\zeta^{10}} \cup \mathcal Q_{\zeta^{23}}, \\
\varphi_{0,\zeta^3,}^{-1}(O_5)&=&\mathcal Q_{\zeta^7}\cup \mathcal Q_{\zeta^{12}} \cup  \mathcal Q_{\zeta^{5}} \cup \mathcal Q_{\zeta^{16}}, \\
\varphi_{0,\zeta^{9},}^{-1}(O_5)&=&\mathcal Q_{\zeta^{13}}\cup \mathcal Q_{\zeta^{18}} \cup \mathcal Q_{\zeta^{11}} \cup \mathcal Q_{\zeta^{22}},\\
\varphi_{0,\zeta^{15},}^{-1}(O_5)&=&\mathcal Q_{\zeta^{19}}\cup \mathcal Q_{1} \cup \mathcal Q_{\zeta^{17}} \cup \mathcal Q_{\zeta^{4}},\\
\varphi_{0,0}(O_5)&=&\mathcal Q_{\zeta^{20}} \cup \mathcal Q_{\zeta^2} \cup \mathcal Q_{\zeta^8} \cup \mathcal Q_{\zeta^{14}}.
    \end{array}
    $$
    Therefore, $\left\{ \varphi_{0,0},
    \varphi_{0,\zeta^{3}}, \varphi_{0,\zeta^{9}}, \varphi_{0,\zeta^{15}}, \varphi_{0,\zeta^{21}}
    \right\}$ is a partial PD set for burst errors impacting positions indexed by $\mathcal Q_b$ for the $b$ listed above.
\end{example}

\begin{proposition}\label{P:2}
Let $C(D, \gamma P_{\infty})$ be a Hermitian one-point code such that $\gamma$ is of one of the forms of Proposition \ref{info_positions} with $i\leq q-2$. For each $a\in\F_{q^2}^*$, fix $b_a$ such that $(a,b_a)\in\mathcal X_q$. Then $C(D, \gamma P_{\infty})$ has a $2$-PD set
$$
\left\{\varphi_{-a,(-a)^{q+1}-b_a+\beta} : a \in \F_{q^2}^*, \beta \in \ker \left( Tr \right) \right\}
$$
of size $q(q^2-1)$.
\end{proposition}

\begin{proof}
    Consider the positions associated with any two affine $\F_{q^2}$-rational points $(a,b)$ and $(c,d)$. If $a=c$, then the result follows from Theorem \ref{P:x_burst}. If $b=d$, then the result follows from Theorem \ref{P:y_burst}. 

    Suppose $a \neq b$ and $c \neq d$. Then 
    $$
    \begin{array}{lcl}
    \varphi_{-a,(-a)^{q+1}-b}(a,b)&=&(0,0)\\
    \varphi_{-a,(-a)^{q+1}-b}(c,d)&=&(c-a,(-a)^qc+d+(-a)^{q+1}-b).
    \end{array}
    $$
    Since $c-a \neq 0$, there exists $\beta \in \ker(Tr)$ so that $(c-a,(-a)^qc+d+(-a)^{q+1}-b +\beta) \in O_q$. As a result, 
    $$
    \begin{array}{lcl}
\varphi_{0,\beta} \circ    \varphi_{-a,(-a)^{q+1}-b}(a,b)&=&(0,\beta) \in O_{q+1}\\
 \varphi_{0,\beta} \circ   \varphi_{-a,(-a)^{q+1}-b}(c,d)&=&(c-a,(-a)^qc+d+(-a)^{q+1}-b + \beta) \in O_q.
    \end{array}
    $$
    Moreover, $\varphi_{-a,(-a)^{q+1}-b+\beta}=\varphi_{0,\beta} \circ   \varphi_{-a,(-a)^{q+1}-b}$ and 
$$
\varphi_{-a,(-a)^{q+1}-b+\beta}\left(\mathcal P_a \cup \left\{(c,d) \right\} \right) \subseteq Q_q \cup Q_{q+1}\cup Q_{q+2}\subseteq \H,
$$
where $\H$ is the set of points on $\mathcal X_q$ associated with check positions of $C(D,\gamma P_{\infty})$.
\end{proof}

Table \ref{T:gamma} lists the values of $\gamma$ that fit one of the forms from Proposition \ref{info_positions} with $i\leq q-2$ for small values of $q$. Theorems \ref{P:x_burst}, \ref{P:y_burst}, \ref{P:2} all apply to Hermitian codes defined by the values for $\gamma$ in Table \ref{T:gamma}.
For $q\geq 7$, there are significantly more values possible for $\gamma$.

\begin{table}[htbp]
    \centering
   \label{T:gamma}
    \begin{tabular}{|c|c|}
        \hline
        $q$ & set of possible $\gamma$ values \\
        \hline
                $3$ & $\{2, 3, 5, 6, 8, 9, 13, 14\}$\\
\hline
$4$ & $\left(\{3,\dots,32\}\backslash\{6, 10, 14, 21, 25, 29\} \right) \cup\{41, 42, 43\}$\\
\hline 
$5$ & $\left( \{4,...,75\}\backslash\{8, 13, 18, 23, 32, 37, 42, 47, 56, 61, 66, 71\} \right) \cup\{91,...,94\}$\\
        \hline
    \end{tabular}
     \caption{Values of $\gamma$ that fit the forms from Proposition 3.2 with $i \leq q-2$ for some small values of $q$, meaning Theorems \ref{P:x_burst}, \ref{P:y_burst}, \ref{P:2} all apply to the one-point code defined with these $\gamma$ values.}
\end{table}

\section{Norm-trace codes} \label{S:nt}
In this section, we consider permutation decoding of one-point norm-trace codes. The one-point Hermitian codes considered in the previous section are special cases of this more general construction. There, a natural ordering of the evaluation points yielded information sets that enabled burst error correction for Hermitian codes. In this section, we consider arbitrary information sets and associated permutation decoding for norm-trace codes. 

\subsection{Automorphisms}
The automorphism group of the curve $\mathcal X_{q,s}$ was computed by Bonini, Montanucci, and Zini in 2020 \cite[Theorem 3.1]{Bonini_nt_autom_20}.
For $s \geq 3$, 
they show
$$
Aut \left(\mathcal X_{q,s} \right) = \left\{ \psi_{\beta, \epsilon}: \epsilon \in \F_{q^s}^*, \beta \in \ker \left( Tr \right) \right\}
$$
where 
$$
\begin{array}{lccl}
\psi_{ \beta, \epsilon}: & x & \mapsto & \epsilon x \\
& y & \mapsto & \epsilon^{\frac{q^s-1}{q-1}}y + \beta.
\end{array}
$$
Thus, 
$$
\mid Aut \left(\mathcal X_{q,s} \right) \mid= q^{2s-1}-q^{s-1}.
$$
Hence, $Aut \left(\mathcal X_{q,s} \right)$ is not transitive for $s \geq 3$ (unlike the $s=2$ case in which the automorphism group of the curve is doubly transitive). However, they also show that 
$$Aut \left(\mathcal X_{q,s} \right)
=
Aut_{D,\gamma P_{\infty}} \left(\mathcal X_{q,s} \right),
$$
so every automorphism of the curve is an automorphism of the code (unlike the $s=2$ case).

Observe that $Aut_{D,\gamma P_{\infty}} \left(\mathcal X_{q,s} \right)$ also fixes the collection of points $(0,\beta)$ for all $\beta \in \ker \left( Tr \right)$ so the approach taken with the Hermitian code where we map to $O_{q+1}$ will not be helpful here.

\begin{remark}
    Notice that $\psi_{\epsilon, \beta}=\varphi_{0,\beta,\epsilon}$. Given $\epsilon \in \ker \left( N \right)$ and $\beta \in \ker \left( Tr \right)$, 
$$
\mid \psi_{\beta,\epsilon} \mid = lcm \left\{ \mid \beta \mid, \mid \epsilon \mid \right\}
$$
where $\mid \beta \mid$ is additive order of $\beta$ and $\mid \epsilon \mid$ is the multiplicative order of $\epsilon$. Moreover, if $q$ is prime, then there exists a code automorphism of large order:
$$
\mid \psi_{\beta,\epsilon} \mid = \frac{q^{s+1}-q}{q-1}.
$$
To see this, we must confirm that there exists $\epsilon \in \ker(N)$ with $|\epsilon|=\frac{q^s-1}{q-1}$ and $\beta \in \ker(Tr)$ with $|\beta|=q$. Indeed, $\mid \ker(Tr) \mid =q$, which is prime, so any nonzero element of $\ker(Tr)$ has order $q$. Let $\F_{q^s}^*=\langle\zeta\rangle$. We can observe that $\zeta^{q^s-1}=(\zeta^{q-1})^{\frac{q^s-1}{q-1}}=1$ implies $\mid \zeta^{q-1}\mid=\frac{q^s-1}{q-1}$, since $\zeta$ is a primitive element of $\F_{q^s}^*$. It then follows that $\psi_{\beta,\epsilon}$ has $q-1$ orbits
$$O_i:=\cup_{j=0}^{q} \mathcal P_{\zeta^{(q-1)j+i}}$$
for $i \in [q-1]$ and $$O_q:=\mathcal P_0.$$
\end{remark}

\subsection{Systematic Representation}
According to \cite[Theorem 3.2]{decreasing_norm_trace_24}, for an $[n,k]$ one-point norm-trace code $C(D, \gamma P_{\infty})$ over $\F_{q^s}$, 
we can specify any collection of $k$ affine $\F_{q^s}$-rational points on $\mathcal X_{q,s}$ to comprise an information set. 

\begin{proposition} \label{arb_pts_info_set}
Consider the one-point norm-trace code $C(D, \gamma P_{\infty})$ of dimension $k$ over $\F_{q^s}$. For $i \in [k]$, take $P_i:=(a_i,b_i) \in \mathcal X_{q,r}$. Then $\left\{ P_1, \dots, P_k \right\}$ is an information set for $C(D, \gamma P_{\infty})$. In fact, if
\begin{equation} \label{E:nt_indicator_fn}
 f_i(x,y):=c_i \left( \frac{x^{q^r-1}+1}{x-a_i}\right)\left( \frac{Tr(y)-Tr(b_i)}{y-b_i}\right)
\end{equation}
where 
$$
c_i:=\begin{cases}
    -1 & \textnormal{if } a_i = 0 \\
    \left(-\frac{q^r-1}{q-1}\right)^{-1} & \textnormal{otherwise,}
\end{cases}
$$
then $$
f_i(P_j)=
\begin{cases}
1 & \textnormal{if } i=j\\
0 & \textnormal{otherwise.}
\end{cases}
$$
and there is a generator matrix $G$ for $C(D, \gamma P_{\infty})$ given by
$$
\begin{array}{lcl}
Row_i(G)&=&\left( f_i(P_1), \dots, f_i(P_k), f_i(P_{k+1}), \dots, f_i(P_n) \right)\\ \ \\
&=&
\left( e_i \mid f_i(P_{k+1}), \dots, f_i(P_n) \right),
\end{array}
$$
for $i \in [k]$.
\end{proposition}

Notice that Proposition \ref{arb_pts_info_set} gives an explicit expression for
 $$G=\left[ I_k \mid A \right] \in \F_{q^r}^{k \times n}$$
and
$$
A=\left[
\begin{array}{ccc}
f_1(P_{k+1}) & \dots &  f_1(P_n) \\
\vdots & & \vdots \\
f_k(P_{k+1}) & \dots &  f_k(P_n) \\
\end{array}
\right] \in \F_{q^r}^{k \times (n-k)}.
$$

\begin{example}\label{E:nt_example}
    Consider the norm-trace curve $\mathcal X_{3,3}: y^9+y^3+y=x^{13}$ over $\F_{27}$, with $\F_{27}^*=\left< \zeta\right>$ where $\zeta^3-\zeta+1=0$. Then the norm and trace taken with respect to the extension $\F_{27}/\F_3$ satisfy
$$N(\zeta^i)=\begin{cases}1 & \textnormal{if } i \textnormal{ is even}\\
2 & \textnormal{if } i \textnormal{ is odd}
\end{cases}
$$
and 
$$Tr(\beta)=\begin{cases}0 & \textnormal{if } \beta = \zeta, \zeta^3, \zeta^9, \zeta^{13}, \zeta^{14}, \zeta^{16}, \zeta^{22}, \zeta^{26}, 0\\
1 & \textnormal{if } \beta =\zeta^5, \zeta^8, \zeta^{15}, \zeta^{17}, \zeta^{19}, \zeta^{20}, \zeta^{23}, \zeta^{24}, \zeta^{25}\\
2 & \textnormal{if } \beta =\zeta^2, \zeta^4, \zeta^6, \zeta^7, \zeta^{10}, \zeta^{11}, \zeta^{12}, a^{18}, \zeta^{21}.
\end{cases}
$$
The $\F_{27}$-rational points of $\mathcal X_{3,3}$ are $(\zeta^i,\beta)$ where 
\begin{enumerate}
    \item 
$i$ is even and $\beta = \zeta^5, \zeta^8, \zeta^{15}, \zeta^{17}, \zeta^{19}, \zeta^{20}, \zeta^{23}, \zeta^{24}, \zeta^{25}$ 
\item $i$ is odd and $\beta = \zeta^2, \zeta^4, \zeta^6, \zeta^7, \zeta^{10}, \zeta^{11}, \zeta^{12}, \zeta^{18}, \zeta^{21}$
\item $a=0$  and $\beta = \zeta, \zeta^3, \zeta^9, \zeta^{13}, \zeta^{14}, \zeta^{16}, \zeta^{22}, \zeta^{26}, 0.$
\end{enumerate}
For each $(\alpha,\beta) \in \mathcal X_{3,3}(\F_{27})$, the indicator function $f_{\alpha,\beta}(x,y)$ such that  
$$f_{\alpha,\beta}(c,d)=\begin{cases}
    1 & \textnormal{if } c=\alpha \textnormal{ and } d=\beta \\
    0& \textnormal{otherwise}
\end{cases}$$
as in Expression (\ref{E:nt_indicator_fn}) is given by 
$$ f_{\alpha,\beta}(x,y):=
2\left( \frac{x^{27}-x}{x-\alpha} \right) \left( \frac{y^9+y^3+y-Tr(\beta)}{y-\beta} \right).
$$
\end{example}

\subsection{Permutation Decoding}
Let $\F_{q^s}^*=\langle \zeta\rangle$, and consider the following automorphism of the norm-trace curve $\mathcal X_{q,s}$, which will play a similar role to the automorphism $\sigma$ of the Hermitian curve as defined in Section \ref{S:herm:systrep}.

$$\begin{array}{lccl}
\psi:=\psi_{ 0, \zeta}: & x & \mapsto & \zeta x \\
& y & \mapsto & \zeta^{\frac{q^s-1}{q-1}}y.
\end{array}$$
$\psi$ partitions the $q^{2s-1}$ affine $\F_{q^s}$-rational points into $q^{s-1}+\frac{q^{s-1}-1}{q-1}+1$ orbits, described below.
\begin{itemize}
    \item $q^{s-1}$ orbits, each with $q^s-1$ points, of the form
    $$\bigcup_{j=1}^{q-1}\mathcal Q_{\zeta^{\ell+j\frac{q^s-1}{q-1}}}$$
    for each $\ell$ such that $(\zeta,\zeta^\ell)\in\mathcal X_{q,s}(\F_{q^s})$,
    \item $\frac{q^{s-1}-1}{q-1}$ orbits, each with $q-1$ points, of the form
    $$\left\{\left(0,\zeta^{i+j\frac{q^s-1}{q-1}}\right):j\in[q-1]\right\}$$
    for each $i$ such that $\zeta^i\in\ker(Tr)$,
    \item $1$ orbit that only contains the point $(0,0)$.
\end{itemize}

\begin{example}
    Consider $\mathcal X_{3,3}$ over $\F_{27}$ as in Example \ref{E:nt_example}. The $\F_{27}$-rational points of $\mathcal X_{3,3}$ are $(\alpha,\beta)$ where 
\begin{enumerate}
    \item $\alpha=\zeta^i$, where $i$ is even, and $\beta = \zeta^5, \zeta^8, \zeta^{15}, \zeta^{17}, \zeta^{19}, \zeta^{20}, \zeta^{23}, \zeta^{24}, \zeta^{25}$,
\item $\alpha=\zeta^i$, where $i$ is odd, and $\beta = \zeta^2, \zeta^4, \zeta^6, \zeta^7, \zeta^{10}, \zeta^{11}, \zeta^{12}, \zeta^{18}, \zeta^{21}$,
\item $\alpha=0$ and $\beta = \zeta, \zeta^3, \zeta^9, \zeta^{13}, \zeta^{14}, \zeta^{16}, \zeta^{22}, \zeta^{26}, 0$.
\end{enumerate}
The orbits of $\psi$ are as follows:
\begin{enumerate}
    \item $9$ orbits, each with $26$ points, of the form $\mathcal Q_{\zeta^i}\cup\mathcal Q_{\zeta^{i+13}}$, where $(\zeta,\zeta^i)\in\mathcal X_{3,3}$ (so $i\in\{2,4,6,7,10,11,12,18,21\}$),
    \item $4$ orbits, each with $2$ points, of the form $\mathcal Q_{\zeta^j}\cup\mathcal Q_{\zeta^{j+13}}=\{(0,\zeta^j),(0,\zeta^{j+13})\}$, for $j\in\{1,3,9,13\}$,
    \item $1$ orbit containing only $(0:0)$.
\end{enumerate}
\end{example}

By \cite[Theorem 3.2]{decreasing_norm_trace_24}, we can choose any $k$ points as an information set for a norm-trace code. In particular, for codes of small enough dimension, we can choose an information set $I$ such that
$$\bigcup_{j=1}^{q-1}\mathcal Q_{\zeta^{\ell+j\frac{q^s-1}{q-1}}}\subseteq[n]\backslash I$$
for some fixed $\ell$ such that $(\zeta:\zeta^\ell)\in\mathcal X_{q,s}(\F_{q^s})$. We then obtain a result equivalent to Theorem \ref{P:y_burst} for norm-trace codes.

\begin{proposition}
    Let $C(D,\gamma P_\infty)$ be an $[n,k]$ one-point norm-trace code over $\F_{q^s}=\langle\zeta\rangle$ such that $k\leq q^{2s-1}-q^s+1.$ Then $\{\psi_{\beta,1}:\beta\in\ker(Tr)\}$ is a partial $\frac{q^s-1}{q-1}$-PD set for errors indexed by $\mathcal Q_b$ for some $b\notin\ker(Tr)$.
\end{proposition}
\begin{proof}
    Fix $\ell$ such that $(\zeta,\zeta^\ell)\in\mathcal X_{q,s}(\F_{q^s})$. Since $k\leq q^{2s-1}-q^s+1$, there exists a set of points $\H$ associated with check positions such that
    $$\bigcup_{j=1}^{q-1}\mathcal Q_{\zeta^{\ell+j\frac{q^s-1}{q-1}}}\subseteq\H.$$
    Fix $b\in\F_{q^s}$ with $b\notin\ker(Tr)$. Then there exists some $\beta\in\ker(Tr)$ such that $b+\beta=\zeta^{\ell+j\frac{q^s-1}{q-1}}$ for some $j\in[q-1]$. It follows that for all $(a,b)\in\mathcal Q_b$,
    $$\psi_{\beta,1}(a,b)=(a,b+\beta)\in\bigcup_{j=1}^{q-1}\mathcal Q_{\zeta^{\ell+j\frac{q^s-1}{q-1}}}.$$
    That is, $\psi_{\beta,1}(\mathcal Q_b)\subseteq\H$. Moreover, $\psi_{\beta,1}$ can correct any error of size $\frac{q^s-1}{q-1}$ that corresponds to positions indexed by
    $$\psi_{\beta,1}^{-1}\left(\bigcup_{j=1}^{q-1}\mathcal Q_{\zeta^{\ell+j\frac{q^s-1}{q-1}}}\right).$$
    Hence, we have a partial $\frac{q^s-1}{q-1}$-PD set.
\end{proof}

\begin{example}
    Consider the norm-trace curve $\mathcal X_{3,3}$ as in Example \ref{E:nt_example}, with $\F_{27}^*=\langle\zeta\rangle$. Recall that $\ker(Tr)=\{0,\zeta,\zeta^3,\zeta^9,\zeta^{13},\zeta^{14},\zeta^{16},\zeta^{22},\zeta^{26}\}$. We designate $\mathcal Q_{\zeta^{2}}\cup\mathcal Q_{\zeta^{15}}$ as containing only check digits. Observe that
    \begin{align*}
        \psi_{0,1}^{-1}(\mathcal Q_{\zeta^{2}}\cup\mathcal Q_{\zeta^{15}})&=\mathcal Q_{\zeta^{2}}\cup\mathcal Q_{\zeta^{15}},\\
        \psi_{\zeta,1}^{-1}(\mathcal Q_{\zeta^{2}}\cup\mathcal Q_{\zeta^{15}})&=\mathcal Q_{\zeta^{4}}\cup\mathcal Q_{\zeta^{23}},\\
        \psi_{\zeta^{3},1}^{-1}(\mathcal Q_{\zeta^{2}}\cup\mathcal Q_{\zeta^{15}})&=\mathcal Q_{\zeta^{18}}\cup\mathcal Q_{\zeta^{24}},\\
        \psi_{\zeta^{9},1}^{-1}(\mathcal Q_{\zeta^{2}}\cup\mathcal Q_{\zeta^{15}})&=\mathcal Q_{\zeta^{7}}\cup\mathcal Q_{\zeta^{19}},\\
        \psi_{\zeta^{13},1}^{-1}(\mathcal Q_{\zeta^{2}}\cup\mathcal Q_{\zeta^{15}})&=\mathcal Q_{\zeta^{21}}\cup\mathcal Q_{\zeta^{25}},\\
        \psi_{\zeta^{14},1}^{-1}(\mathcal Q_{\zeta^{2}}\cup\mathcal Q_{\zeta^{15}})&=\mathcal Q_{\zeta^{10}}\cup\mathcal Q_{\zeta^{17}},\\
        \psi_{\zeta^{16},1}^{-1}(\mathcal Q_{\zeta^{2}}\cup\mathcal Q_{\zeta^{15}})&=\mathcal Q_{\zeta^{11}}\cup\mathcal Q_{\zeta^{5}},\\
        \psi_{\zeta^{22},1}^{-1}(\mathcal Q_{\zeta^{2}}\cup\mathcal Q_{\zeta^{15}})&=\mathcal Q_{\zeta^{6}}\cup\mathcal Q_{\zeta^{20}},\\
        \psi_{\zeta^{26},1}^{-1}(\mathcal Q_{\zeta^{2}}\cup\mathcal Q_{\zeta^{15}})&=\mathcal Q_{\zeta^{12}}\cup\mathcal Q_{\zeta^{8}}.
    \end{align*}
Then, we can correct any burst errors associated with the points in the sets on the right-hand side of the list above.
\end{example}

\section{Conclusion} \label{S:concl}
In this paper, we have established permutation decoding that corrects specific error patterns in one-point Hermitian and norm-trace codes. To our knowledge, this is the first work demonstrating permutation decoding for algebraic geometry codes on curves of positive genus. It would be interesting to determine the automorphism groups of other families of algebraic geometry codes, such as those arising from Kummer extensions or Castle curves, and to explore their potential for permutation decoding. 

\bibliographystyle{plain}
\bibliography{bib_glm}

\end{document}